\documentclass[10pt,journal]{IEEEtran}

\usepackage{cite}

\usepackage{graphicx}
\usepackage{epstopdf}
\graphicspath{{./figures}}
\DeclareGraphicsExtensions{.pdf,.eps}

\usepackage{amsmath,amssymb,amsthm,amsfonts,amscd,amsbsy,amsxtra}
\usepackage{mathtools}
\usepackage{siunitx}
\usepackage{algorithm}
\usepackage{algpseudocode}
\usepackage{float}                
\floatstyle{ruled}\restylefloat{algorithm}

\algnewcommand\Input{\State\textbf{input: }}
\algnewcommand\Output{\State\textbf{output: }}
\algnewcommand\Initialize{\State\textbf{initialize: }}
\algnewcommand\CmtLine[1]{\Statex \hspace{\algorithmicindent}{\texttt{/* #1 */}}}

\usepackage[caption=false,font=footnotesize]{subfig}

\newcommand*{\herm}{\mathsf{H}}
\newcommand{\E}{\mathbb{E}}
\newcommand{\C}{\mathbb{C}}
\DeclareMathOperator{\tr}{tr}
\DeclareMathOperator*{\argmax}{arg\,max}

\DeclareMathOperator{\cov}{Cov}
\DeclarePairedDelimiter\abs{\lvert}{\rvert}%
\DeclarePairedDelimiter\norm{\lVert}{\rVert}%

\newtheorem{proposition}{Proposition}

\IEEEoverridecommandlockouts

\begin{document}
\bstctlcite{IEEEexample:BSTcontrol}

\title{OAMP-Aided Joint Channel Estimation and Data Detection for ODDM Systems}

\author{
    Kehan~Huang,~\IEEEmembership{Graduate~Student~Member,~IEEE,}
    Min~Qiu,~\IEEEmembership{Senior~Member,~IEEE,}
    Akram~Shafie,~\IEEEmembership{Member,~IEEE,}
    and~Jinhong~Yuan,~\IEEEmembership{Fellow,~IEEE}

    \thanks{The work of K. Huang, M. Qiu, A. Shafie, and J. Yuan was supported in part by the Australian Research Council (ARC) Discovery Project under Grant DP220103596, in part by the ARC Linkage Project under Grant LP200301482, and in part by the Connectivity Innovation Network, Australia. The work of M. Qiu was also supported by the SJTU-ExploreX Funding under Grant SD6040004/153. \emph{(Corresponding author: Min Qiu.)}}

    \thanks{K. Huang, A. Shafie, and J. Yuan are with the School of Electrical Engineering and Telecommunications, University of New South Wales (UNSW), Sydney, NSW 2052, Australia (email: \{kehan.huang, akram.shafie, j.yuan\}@unsw.edu.au). M. Qiu is with the Global College, Shanghai Jiao Tong University, Shanghai 200240, China (email: min\_qiu@sjtu.edu.cn).}
}

\markboth{Accepted for publication in IEEE Transactions on Vehicular Technology.}{Accepted for publication in IEEE Transactions on Vehicular Technology.}

\maketitle

\begin{abstract}
    In this work, to address the challenge of joint channel estimation and data detection (JED) for orthogonal delay-Doppler (DD) division multiplexing (ODDM) in doubly selective channels, we propose an orthogonal approximate message passing (OAMP)-aided JED (OAMP-JED) receiver. We first formulate a bilinear cross-domain JED model, which can be linearized into separate channel estimation and data detection subproblems. The proposed OAMP-JED receiver alternately executes two OAMP modules for these subproblems, effectively coupled through a variational noise term to account for model uncertainty. Leveraging OAMP's error orthogonality, we derive closed-form scalar-variance updates to enable efficient and principled soft information exchange between the modules, thereby mitigating error propagation during JED. Simulation results show that, for both uncoded and coded ODDM, OAMP-JED achieves a lower bit error rate (BER) than benchmark schemes. Moreover, its BER performance closely approaches that of OAMP with perfect CSI.
\end{abstract}

\section{Introduction}\label{sec:intro}

Emerging vehicular-to-everything (V2X) and low Earth orbit (LEO) satellite links exhibit significant selectivity in both time and frequency. This poses challenges for the widely adopted orthogonal frequency-division multiplexing (OFDM) systems. As an alternative, the orthogonal time-frequency space (OTFS) modulation \cite{Hadani2017OTFS_WCNC} was proposed to harness channel diversity in the delay-Doppler (DD) domain, and it has since motivated waveform designs that exploit the DD-domain channel representation. Taking inspiration from OTFS, \cite{Lin2022OrthogonalModulation} recently introduced the orthogonal DD division multiplexing (ODDM) modulation. The basis function construction in ODDM explicitly matches the inherent operations of doubly selective channels, providing a coupling between the DD spreading function and a low-dimensional set of physical parameters.

Despite the promising features of ODDM/OTFS, the DD spreading in doubly selective channels induces inter-symbol interference (ISI), which significantly complicates receiver design. In particular, the pronounced ISI motivates iterative, nonlinear data detection algorithms to strike a balance between error performance and computational complexity \cite{Raviteja2018MPA,Thaj2020LowSystems,Li2022CrossModulation,Huang2024ODDM_Performance,Li2025OTFS_ParallelCoding}. Among these algorithms, orthogonal approximate message passing (OAMP) is a promising candidate, known for its asymptotic optimality for linear observation models with right-unitarily invariant (RUI) measurement matrices \cite{Ma2017OAMP}. Although the channel matrices of ODDM/OTFS are not RUI in general, \cite{Wen2022OAMP_OTFS} has empirically demonstrated the superior performance of OAMP compared to conventional detectors in \cite{Raviteja2018MPA,Thaj2020LowSystems,Li2022CrossModulation}.

Standalone detection algorithms assume perfect channel state information (CSI) at the receiver \cite{Raviteja2018MPA,Thaj2020LowSystems,Li2022CrossModulation,Wen2022OAMP_OTFS,Li2025OTFS_ParallelCoding}. In point-to-point links, this corresponds to a two-stage pipeline: the CSI is first estimated from known pilot symbols and then treated as known for subsequent data detection. This pipeline requires receiver-side separation between pilot and data symbols to prevent mutual interference from degrading each stage, often achieved by guard intervals \cite{Raviteja2019EmbeddedChannels,Huang2024ODDM_Performance}. Considering the 2D DD spreading in doubly selective channels, guard intervals can be very expensive in terms of spectral efficiency. To address this limitation, joint channel estimation and data detection (JED) receivers \cite{Mishra2022OTFS_SuperimposedPilots,Wang2022OTFS_JED_VBI,Wen2024MFOAMP,Yuan2021DataAided_ChanEst} have been proposed, enabling superimposed pilots with minimal overhead.

In this paper, we investigate advanced JED for ODDM with a superimposed pilot. We focus on a specific type of JED receivers \cite{Mishra2022OTFS_SuperimposedPilots,Wang2022OTFS_JED_VBI,Wen2024MFOAMP} that actively use detected data symbols as soft priors to iteratively improve channel estimation, which can provide higher signal-to-noise ratio (SNR) gain compared to the interference cancellation-based JED receivers \cite{Yuan2021DataAided_ChanEst}. For example, the superimposed pilot-iterative (SP-I) receiver in \cite{Mishra2022OTFS_SuperimposedPilots} alternates between a soft linear minimum mean square error (MMSE) (LMMSE) channel estimator and a message passing algorithm detector. However, SP-I relies on prior knowledge of the channel support to achieve good performance, which may not be available in practice. The Bayesian inference receiver in \cite{Wang2022OTFS_JED_VBI} provides a principled JED method under reasonable prior assumptions, but it requires a nested data-detection loop and leads to prohibitive computational complexity for typical frame sizes. Leveraging the powerful OAMP framework for both channel estimation and data detection, \cite{Wen2024MFOAMP} proposed a mean-field (MF)-OAMP receiver that iteratively updates measurement matrices by MF-aided rowwise scaling, which, however, can bias the observation model.

Motivated by OAMP's effectiveness for detection and the limitations of existing JED receivers, we propose the OAMP-aided JED (OAMP-JED) receiver for ODDM to achieve superior BER performance. Our main contributions are as follows:

\textbf{1}) We establish a bilinear JED model for doubly selective channels. By conditionally linearizing this model with respect to the channel and data variables, we derive the corresponding OAMP updates for the channel-estimation and data-detection subproblems.

\textbf{2}) We propose a novel OAMP-JED receiver that alternates between two OAMP modules for channel estimation and data detection, effectively coupled through a variational noise term to account for model uncertainty. Importantly, we derive closed-form scalar-variance updates to characterize the noise, enabling efficient and reliable soft information exchange.

\textbf{3}) We provide extensive simulation results and show that the proposed OAMP-JED receiver can closely approach the BER performance of OAMP with perfect CSI for both uncoded and coded ODDM systems.

\textbf{Notations}: $(\cdot)^*$ and $(\cdot)^\herm$ denote the complex conjugate and Hermitian transpose, respectively. $\boldsymbol{F}_N$ denotes the normalized $N$-point DFT matrix. $\delta(\cdot)$ denotes the Dirac delta function and $\delta_{\cdot,\cdot}$ denotes the Kronecker delta. $\norm{\boldsymbol{A}}$ outputs the 2-norm of matrix $\boldsymbol{A}$. $\mathcal{A}\times \mathcal{B}$ gives the Cartesian product of two sets.

\section{System Model}\label{sec:sysmod}

\subsection{Transmitter}
We consider a data sequence $\boldsymbol{x}_\mathrm{d}\in\mathcal{A}^{N_s}$ and a superimposed pilot sequence $\boldsymbol{x}_\mathrm{p}\in\C^{N_s}$, where $\mathcal{A}\subset\C$ denotes the constellation alphabet.\footnote{
    We allocate pilots in the data domain for notational simplicity. However, the proposed scheme can be straightforwardly extended to time-domain pilots.
} This leads to the transmit sequence
\begin{align}
    \boldsymbol{x}=\boldsymbol{x}_\mathrm{d}+\boldsymbol{x}_\mathrm{p},
\end{align}
where we define the pilot-to-data energy ratio as $\gamma \triangleq E_p/E_d$, with $E_d \triangleq \frac{1}{N_s}\E\bigl[\norm{\boldsymbol{x}_\mathrm{d}}^2\bigr]$ and $E_p \triangleq \frac{1}{N_s}\E\bigl[\norm{\boldsymbol{x}_\mathrm{p}}^2\bigr]$.

We then consider a unitary modulation scheme
\begin{align}
    \boldsymbol{s} = \boldsymbol{F}_{\!\mathrm{mod}}\boldsymbol{x},
\end{align}
where $\boldsymbol{s}\in\C^{N_s}$ and $\boldsymbol{F}_{\!\mathrm{mod}}\in\C^{N_s\times N_s}$ are the time-domain modulated sequence and the modulation matrix, respectively. We adopt the zero-padded (ZP) frame structure \cite{Thaj2020LowSystems} for transmission, where $\boldsymbol{s}$ is segmented into $N$ shorter sequences of length $M_d$ such that $N_s = M_dN$. Each segment is padded with $M_0$ zeros to get a length of $M=M_d+M_0$. The time-domain transmitted sequence $\boldsymbol{s}_{\mathrm{zp}}\in\C^{MN}$ is expressed as
\begin{align}
    \boldsymbol{s}_{\mathrm{zp}}=\boldsymbol{F}_{\mathrm{zp}}\boldsymbol{s},
\end{align}
with
\begin{align}
    \boldsymbol{F}_{\mathrm{zp}} = \boldsymbol{I}_N \otimes
    \begin{bmatrix}
        \boldsymbol{I}_{M_d} \\
        \boldsymbol{0}_{M_0\times M_d}
    \end{bmatrix}
    \in\{0,1\}^{N_r\times N_s}.
\end{align}

\begin{figure}[t]
    \centering
    \subfloat[DD-domain representation.]{\label{fig:frame_dd}\includegraphics[scale=0.52,trim={0 0 -30 0},clip]{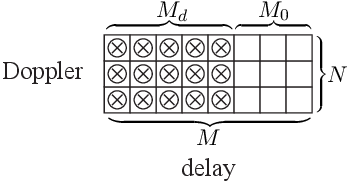}}
    \\
    \centering
    \subfloat[Time-domain representation.]{\label{fig:frame_time}\includegraphics[scale=0.52,trim={0 0 0 0},clip]{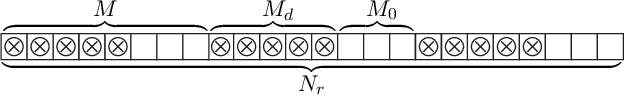}}
    \\
    \centering
    \subfloat{\includegraphics[scale=0.52,trim={0 0 0 0},clip]{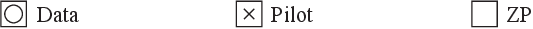}}
    \caption{Frame structure of $\boldsymbol{s}_{\mathrm{zp}}$ with $M=8,N=3,M_0=3$.}
    \label{fig:szp}
\end{figure}

When $\boldsymbol{s}_{\mathrm{zp}}$ is transmitted with a root-Nyquist pulse, it results in the ODDM waveform \cite{Lin2022OrthogonalModulation}. In that case, $\boldsymbol{x}$ is the DD-domain symbol vector with $\boldsymbol{s}$ being its time-domain representation. We transmit $\boldsymbol{s}_{\mathrm{zp}}$ over bandwidth $\frac{M}{T}$ for duration $NT$, yielding delay resolution $\Delta\tau=\frac{T}{M}$ and Doppler resolution $\Delta\nu = \frac{1}{NT}$. The corresponding ODDM modulation matrix is\footnote{
    The proposed scheme applies to general unitary modulation matrices.
}
\begin{align}
    \boldsymbol{F}_{\!\mathrm{mod}} = \boldsymbol{P}\left(\boldsymbol{I}_{M_d}\otimes\boldsymbol{F}_N^\herm\right),
\end{align}
where $\boldsymbol{P}\in\{0,1\}^{{M_d}N\times{M_d}N}$ is a permutation matrix
\begin{align}
    (\boldsymbol{P})_{i,j}=
    \begin{cases}
        0, & i=n{M_d}+m,j=mN+n,
        \\
        1, & \text{otherwise},
    \end{cases}
\end{align}
for $m=0,\dots,{M_d}-1$ and $n=0,\dots,N-1$. The frame structure of $\boldsymbol{s}_{\mathrm{zp}}$ and its DD-domain representation are illustrated in Fig. \ref{fig:szp}.

\subsection{Doubly Selective Channel}\label{sec:ddchannel}

Consider a doubly selective channel with $P$ resolvable paths, characterized by the DD-domain spreading function: $h(\tau,\nu)=\sum_{p=1}^{P} h_p \delta (\tau-\tau_p)\delta (\nu-\nu_p)$, where $h_p$, $\tau_p$, and $\nu_p$ are the channel gain, delay shift, and Doppler shift of the $p$-th path, respectively \cite{Lin2022OrthogonalModulation}. Let $l_p$ and $k_p$ denote the normalized delay and Doppler shifts, so that $\tau_p=l_p\Delta\tau$ and $\nu_p=k_p\Delta\nu$. Also denote the sets of normalized delay and Doppler shifts by $\mathcal{L}=\{0,\dots,l_{\max}\}$ and $\mathcal{K}=\{-k_{\max},\dots,k_{\max}\}$, with maximum delay and Doppler shifts $l_{\max}$ and $k_{\max}$. Here we assume on-grid delay and Doppler shifts, i.e., $l_p,k_p\in\mathbb{Z}$. To cover all potential channel responses, we predefine sufficiently large $l_{\max}$ and $k_{\max}$, and assume $l_p\leq l_{\max}$ and $\abs{k_p}\leq k_{\max}$ for $p=1,\dots,P$. Then, the DD channel can be rewritten as:
\begin{equation}\label{ddChan_int}
    h[l,k] = \sum\nolimits_{p=1}^{P} h_p\,\delta_{l,l_p}\,\delta_{k,k_p}, \quad (l,k)\in\mathcal{L}\times\mathcal K.
\end{equation}
Accordingly, we set the zero padding length to $M_0=l_{\max}$.

Notably, this paper focuses on the on-grid channel model. In practice, off-grid delay and Doppler shifts introduce model mismatch, which may degrade the channel estimation module in Section~\ref{sec:modA}. Such mismatch can be mitigated using established techniques, such as grid refinement \cite{Shan2025ODDM_GridRefinement} and Taylor-series approximation \cite{Shi2024BiVAMP_OTFS_ChanEst}, which we leave for future work.

\subsection{Receiver}

The waveform passed through the channel is then matched-filtered with the same root-Nyquist pulse at the receiver. After sampling at $t=q\frac{T}{M}$ for $q=0,\dots,N_r-1$, the $q$-th element of time-domain received vector $\boldsymbol{r}\in\C^{N_r}$ is given by \cite{Lin2022OrthogonalModulation}
\begin{equation}\label{eq:io_time}
    r[q] = \sum_{l\in\mathcal{L},l\leq q}\sum_{k\in\mathcal{K}} h[l,k] e^{j2\pi\frac{k(q-l)}{MN}} s_{\mathrm{zp}}\bigl[q-l\bigr] + w[q],
\end{equation}
where $w[q]\sim\mathcal{CN}\left(0,\sigma_w^2\right)$ is the sampled AWGN. The observation length is $N_r=MN$. Then, the effective time-domain input-output (IO) relation for $\boldsymbol{s}$ can be written as
\begin{align}\label{eq:io_time_mat}
    \boldsymbol{r}=\boldsymbol{G}\boldsymbol{s}+\boldsymbol{w},
\end{align}
with $\boldsymbol{w}\sim\mathcal{CN}\left(\boldsymbol{0},\sigma_w^2\boldsymbol{I}\right)$ and $\boldsymbol{G}=\tilde{\boldsymbol{G}}\boldsymbol{F}_{\mathrm{zp}}$, where $\tilde{\boldsymbol{G}}\in\C^{N_r\times N_r}$ has its $(q_1,q_2)$-th element being
\begin{align}\label{eq:G}
    [\tilde{\boldsymbol{G}}]_{q_1,q_2} {=}
    \begin{cases}
        \sum_{k\in\mathcal{K}} h[q_1{-}q_2,k] e^{j2\pi\frac{kq_2}{MN}}, & q_1{-}q_2\in\mathcal{L},
        \\
        0, & \text{otherwise}.
    \end{cases}
\end{align}

In this work, we consider the JED problem, where $\boldsymbol{G}$ is unknown and needs to be estimated along with $\boldsymbol{s}$, making \eqref{eq:io_time_mat} bilinear. Therefore, we can reparameterize \eqref{eq:io_time_mat} to be explicitly linear in a channel vector $\boldsymbol{h} \triangleq \left[h[l,k]\right]_{(l,k)\in\mathcal{L}\times\mathcal{K}}\in\mathbb C^{N_h}$ for channel estimation, where $N_h=\abs{\mathcal{L}}\abs{\mathcal{K}}$ is the maximum number of channel taps. Specifically, we write
\begin{align}\label{eq:rBh}
    \boldsymbol{r}=\boldsymbol{B}\boldsymbol{h}+\boldsymbol{w},
\end{align}
where
\begin{equation}\label{eq:B}
\boldsymbol{B} \triangleq \left[\boldsymbol{b}_{l,k}(\boldsymbol{s})\right]_{(l,k)\in\mathcal{L}\times\mathcal{K}}\in\C^{N_r\times N_h}
\end{equation}
is the signal matrix with each signal basis $\boldsymbol{b}_{l,k}(\boldsymbol{s})\in\C^{N_r}$ stacked as columns. The $q$-th element of $\boldsymbol{b}_{l,k}(\boldsymbol{s})$ is $[\boldsymbol{b}_{l,k}(\boldsymbol{s})]_q = e^{j2\pi\frac{k(q-l)}{MN}} s_{\mathrm{zp}}[q-l]$, for $q\geq l$ and $q\in\{q=0,\dots,N_r-1\}$.

\section{OAMP-Aided Joint Channel Estimation and Data Detection over Doubly Selective Channels}\label{sec:oamp-jed}

In this section, we propose the OAMP-JED receiver for ODDM over doubly selective channels. To facilitate derivation, we first introduce two OAMP modules based on conditional linearized models \cite{Ma2017OAMP}. Module~A performs channel estimation conditioned on the current soft data estimate, whereas Module~B performs data detection conditioned on the current channel estimate. The two modules are then coupled through principled scalar-variance updates to form the complete OAMP-JED receiver, where the model uncertainty induced by conditional linearization is absorbed into an effective white-noise term.

\subsection{OAMP for Channel Estimation (Module~A)}\label{sec:modA}
Based on \eqref{eq:rBh}, we first consider the following linearized observation model for channel estimation:
\begin{align}\label{eq:rBhv}
    \boldsymbol{r}=\boldsymbol{B}\boldsymbol{h}+\boldsymbol{v},
\end{align}
where $\boldsymbol{B}$ is conditioned on the current soft data estimate. We introduce a variational effective noise term $\boldsymbol{v}$ to account for the model uncertainty introduced by this conditioning and AWGN, which can be modeled as Gaussian due to the central limit theorem. In general iterative receivers, $\boldsymbol{v}$ is colored due to error propagation \cite{Huang2024ODDM_Performance}. In the proposed OAMP-JED framework, however, $\boldsymbol{v}$ can be approximated as a white Gaussian term, i.e., $\boldsymbol{v}\sim\mathcal{CN}(\boldsymbol{0},\sigma_v^2\boldsymbol{I})$, as detailed in Section~\ref{sec:JEDcoupling}.

The OAMP algorithm solves such a linear system by evolving between a decorrelated linear estimator (LE) and a divergence-free nonlinear estimator (NLE) \cite{Ma2017OAMP}. Starting with prior channel vector mean $\tilde{\boldsymbol{h}}_\alpha^{(0)}$, the OAMP update of Module~A at iteration $\ell$ of the proposed algorithm is written as
\begin{subequations}\label{eq:oamp_ce}
\begin{align}
&\text{LE:} &\tilde{\boldsymbol{h}}_\beta^{(\ell)} &= f_h\Bigl(\tilde{\boldsymbol{h}}_\alpha^{(\ell)}\Bigr),\qquad\qquad
\label{eq:oamp_ce_le}\\
&\text{NLE:} &\tilde{\boldsymbol{h}}_\alpha^{(\ell+1)} &= \eta_h\Bigl(\tilde{\boldsymbol{h}}_\beta^{(\ell)}\Bigr).
\label{eq:oamp_ce_nle}
\end{align}
\end{subequations}
Here, the subscript $h$ indicates functions for the estimation of $\boldsymbol{h}$. For notational brevity, we drop the iteration index $\ell$ in the following. The LE
\begin{align}\label{eq:fh_le}
    f_h\Bigl(\tilde{\boldsymbol{h}}_\alpha\Bigr)=\tilde{\boldsymbol{h}}_\alpha + \boldsymbol{W}_{\!\!h}\Bigl(\boldsymbol{r} - \boldsymbol{B} \tilde{\boldsymbol{h}}_\alpha\Bigr)
\end{align}
is subject to the decorrelated constraint $\tr(\boldsymbol{I}\!\!-\!\!\boldsymbol{W}_{\!\!h}\boldsymbol{B})\!\!=\!\!0$, while the NLE $\eta_h\Bigl(\tilde{\boldsymbol{h}}_\beta\Bigr)$ is a component-wise Lipschitz continuous function subject to the divergence-free constraint $\frac{1}{N_h}\sum_{i=0}^{N_h-1}\frac{\partial\eta\left(\tilde{h}_\beta[i]\right)}{\partial\tilde{h}_\beta[i]}{=}0$. It has been proven that \eqref{eq:oamp_ce} ensures orthogonality between the input and output errors of LE and NLE if $\boldsymbol{B}$ is RUI \cite{Ma2017OAMP}. This property admits accurate characterization of OAMP's asymptotic error performance via recursive scalar mean squared error (MSE) updates: $\phi_h^2{=}\E\bigl[(\tilde{\boldsymbol{h}}_\beta{-}\boldsymbol{h})^2\bigr]$ and  $\chi_h^2{=}\E\bigl[(\tilde{\boldsymbol{h}}_\alpha{-}\boldsymbol{h})^2\bigr]$, namely state evolution (SE).

With error orthogonality, one can design the optimal LE and NLE to minimize the MSE under the decorrelated and divergence-free constraints. The optimal LE is a scaled LMMSE estimator with \cite{Ma2017OAMP}
\begin{align}\label{eq:WHat_B}
    \boldsymbol{W}_{\!\!h}{=}\frac{N_h}{\tr\Bigl(\hat{\boldsymbol{W}}_{\!\!h}\boldsymbol{B}\Bigr)}\hat{\boldsymbol{W}}_{\!\!h}, \quad\, \hat{\boldsymbol{W}}_{\!\!h}{=}\boldsymbol{B}^\herm\biggl(\!\boldsymbol{B}\boldsymbol{B}^\herm{+}\frac{\sigma_v^2}{\chi_h^2}\boldsymbol{I}_{N_r}\!\biggr)^{-1},
\end{align}
whose MSE can be approximated as
\begin{align}\label{eq:phi2_h}
    \phi_h^2 {=} \frac{1}{N_h}\!\Bigl(\chi_h^2\tr\Bigl(\!(\boldsymbol{I}{-}\boldsymbol{W}_{\!\!h}\boldsymbol{B})(\boldsymbol{I}{-}\boldsymbol{W}_{\!\!h}\boldsymbol{B})^\herm\!\Bigr){+}\sigma_v^2\tr\left(\boldsymbol{W}_{\!\!h}\boldsymbol{W}_{\!\!h}^\herm\right)\Bigr).
\end{align}
The optimal NLE is given by \cite{Ma2017OAMP}
\begin{align}\label{eq:etah}
    \eta_h(\tilde{\boldsymbol{h}}_\beta) &= \frac{\phi_h^2}{\phi_h^2{-}mse(\phi_h^2)}\left(\hat{\eta}_h\left(\tilde{\boldsymbol{h}}_\beta\right)-\frac{mse(\phi_h^2)}{\phi_h^2}\tilde{\boldsymbol{h}}_\beta\right),
\end{align}
where $\hat{\eta}_h\bigl(\tilde{\boldsymbol{h}}_\beta\bigr)=\E\bigl[\boldsymbol{h}|\tilde{\boldsymbol{h}}_\beta=\boldsymbol{h}+\phi_h\boldsymbol{z},\boldsymbol{z}{\sim}\mathcal{CN}(\boldsymbol{0},\boldsymbol{I})\bigr]$ is the MMSE denoiser and $mse\left(\phi_h^2\right)$ is the posterior MSE. To account for the sparsity of doubly selective channels, the MMSE denoiser uses a Bernoulli-Gaussian prior for $\boldsymbol{h}$ \cite{Wen2024MFOAMP}. The corresponding MSE for NLE is approximated as \cite{Ma2017OAMP}
\begin{align}\label{eq:chi2_h}
    \chi_h^2 = \frac{\bigl\lVert\boldsymbol{r}-\boldsymbol{B}\tilde{\boldsymbol{h}}_\alpha\bigr\rVert^2-N_r\sigma_v^2}{\tr\bigl(\boldsymbol{B}^\herm\boldsymbol{B}\bigr)}.
\end{align}

\subsection{Cross-Domain OAMP for Data Detection (Module~B)}\label{sec:modB}
Similar to Module~A, we consider the following linearized observation model for data detection based on \eqref{eq:io_time_mat}:
\begin{align}\label{eq:rGsu}
    \boldsymbol{r}=\boldsymbol{G}\boldsymbol{s}+\boldsymbol{u},
\end{align}
where $\boldsymbol{G}$ is conditioned on the current channel estimate and $\boldsymbol{u}\sim\mathcal{CN}(\boldsymbol{0},\sigma_u^2\boldsymbol{I})$ denotes the corresponding variational effective noise term. A low-complexity cross-domain OAMP (CD-OAMP) step is formulated for data detection over \eqref{eq:rGsu}. Starting with prior data-domain signal mean $\tilde{\boldsymbol{x}}_\alpha^{(0)}$, the CD-OAMP update of Module~B at iteration $\ell$ can be written as
\begin{subequations}\label{eq:oamp_dd}
\begin{align}
&\text{LE:} &\tilde{\boldsymbol{s}}_\beta^{(\ell)} &= f_s\Bigl(\tilde{\boldsymbol{s}}_\alpha^{(\ell)}\Bigr),\ \tilde{\boldsymbol{s}}_\alpha^{(\ell)}=\boldsymbol{F}_{\!\mathrm{mod}}\tilde{\boldsymbol{x}}_\alpha^{(\ell)},
\label{eq:oamp_dd_le}\\
&\text{NLE:} &\tilde{\boldsymbol{x}}_\alpha^{(\ell+1)} &= \eta_s\Bigl(\tilde{\boldsymbol{x}}_\beta^{(\ell)}\Bigr),\ \tilde{\boldsymbol{x}}_\beta^{(\ell)}=\boldsymbol{F}_{\!\mathrm{mod}}^\herm\tilde{\boldsymbol{s}}_\beta^{(\ell)},
\label{eq:oamp_dd_nle}
\end{align}
\end{subequations}
Here, the subscript $s$ indicates functions for the estimation of $\boldsymbol{s}$, and $f_s(\cdot)$ follows the same update structure as in \eqref{eq:fh_le}.

Accordingly, the optimal LE adopts a scaled LMMSE matrix given by
\begin{align}\label{eq:WHat_G}
    \boldsymbol{W}_{\!\!s}{=}\frac{N_s}{\tr(\hat{\boldsymbol{W}}_{\!\!s}\boldsymbol{G})}\hat{\boldsymbol{W}}_{\!\!s},\ \ \,\hat{\boldsymbol{W}}_{\!\!s}{=}\boldsymbol{G}^\herm\biggl(\!\boldsymbol{G}\boldsymbol{G}^\herm{+}\frac{\sigma_u^2}{\chi_s^2}\boldsymbol{I}_{N_r}\!\biggr)^{-1}.
\end{align}
The MSE of the LE message can be approximated as
\begin{align}\label{eq:phi2_s}
    \phi_s^2 {=} \frac{1}{N_s}\!\Bigl(\chi_s^2\tr\Bigl(\!(\boldsymbol{I}{-}\boldsymbol{W}_{\!\!s}\boldsymbol{G})(\boldsymbol{I}{-}\boldsymbol{W}_{\!\!s}\boldsymbol{G})^\herm\!\Bigr){+}\sigma_u^2\tr\bigl(\boldsymbol{W}_{\!\!s}\boldsymbol{W}_{\!\!s}^\herm\bigr)\Bigr).
\end{align}
The optimal NLE has the same form as \eqref{eq:etah}, except that the MMSE denoiser uses data-dependent priors. For uncoded ODDM, $\hat{\eta}_s\left(\tilde{\boldsymbol{x}}_\beta\right)$ and $mse(\phi_s^2)$ are the posterior mean and average posterior variance of symbols after symbol-wise detection under constellation $\mathcal{A}$. For coded ODDM, they are computed from the decoder output as in \cite{Li2025OTFS_ParallelCoding}. The MSE of the NLE message can be approximated as
\begin{align}\label{eq:chi2_s}
    \chi_s^2 = \frac{\bigl\lVert\boldsymbol{r}-\boldsymbol{G}\tilde{\boldsymbol{s}}_\alpha\bigr\rVert^2-N_r\sigma_u^2}{\tr\bigl(\boldsymbol{G}^\herm\boldsymbol{G}\bigr)}.
\end{align}

\subsection{Bayesian Noise Inflation for JED}\label{sec:JEDcoupling}

Now we couple Module~A and Module~B and propose the OAMP-JED receiver to jointly estimate both the channel vector $\boldsymbol{h}$ and the signal vector $\boldsymbol{s}$. While modules A and B operate with white effective noises $\boldsymbol{v}$ and $\boldsymbol{u}$ to comply with the OAMP framework, we begin the derivation from more general Gaussian noise models $\boldsymbol{v}\sim\mathcal{CN}(\boldsymbol{0},\boldsymbol{\Sigma}_v)$ and $\boldsymbol{u}\sim\mathcal{CN}(\boldsymbol{0},\boldsymbol{\Sigma}_u)$ to capture the colored uncertainty induced by conditional linearization. We first derive the variational noise term $\boldsymbol{v}$ in \eqref{eq:rBhv} for Module~A. The optimal choice of $\boldsymbol{\Sigma}_v$ is
\begin{align}
    \boldsymbol{\Sigma}_v{=}\argmax_{\boldsymbol{\Sigma}_v} p(\boldsymbol{r}\mid\boldsymbol{\Sigma}_v){=}\argmax_{\boldsymbol{\Sigma}_v} \int p(\boldsymbol{r},\boldsymbol{h}\mid\boldsymbol{\Sigma}_v)d\boldsymbol{h},
\end{align}
which, however, is intractable since $\boldsymbol{v}=\boldsymbol{r}-\boldsymbol{B}\boldsymbol{h}$ is correlated with $\boldsymbol{h}$, resulting in a mixture model.

Inspired by \cite{Moon1996EM}, we treat $\boldsymbol{h}$ as a latent variable and learn $\boldsymbol{\Sigma}_v$ by maximizing its expected complete-data log-likelihood. We first write the log-likelihood function
\begin{align}
    \log p(\boldsymbol{r}\mid\boldsymbol{\Sigma}_v)&=\log p(\boldsymbol{r},\boldsymbol{h}\mid\boldsymbol{\Sigma}_v)-\log p(\boldsymbol{h}\mid\boldsymbol{r},\boldsymbol{\Sigma}_v).
\end{align}
Then, we take expectation of both sides over $p\left(\boldsymbol{h}\middle|\boldsymbol{r},\boldsymbol{\Sigma}_v^{(\ell)}\right)$ under the last-iteration estimate $\boldsymbol{\Sigma}_v^{(\ell)}$, yielding
\begin{align}\label{eq:llQH_v}
    \log p(\boldsymbol{r}\mid\boldsymbol{\Sigma}_v)=Q\left(\boldsymbol{\Sigma}_v\middle|\boldsymbol{\Sigma}_v^{(\ell)}\right)+H\left(\boldsymbol{\Sigma}_v\middle|\boldsymbol{\Sigma}_v^{(\ell)}\right),
\end{align}
where
\begin{align}
    Q\left(\boldsymbol{\Sigma}_v\middle|\boldsymbol{\Sigma}_v^{(\ell)}\right){=}\int p\left(\boldsymbol{h}\middle|\boldsymbol{r},\boldsymbol{\Sigma}_v^{(\ell)}\right) \log p(\boldsymbol{r},\boldsymbol{h}\mid\boldsymbol{\Sigma}_v)d\boldsymbol{h}
\end{align}
is the expected complete-data log-likelihood, and
\begin{align}
    H\left(\boldsymbol{\Sigma}_v\middle|\boldsymbol{\Sigma}_v^{(\ell)}\right){=}-\!\!\int p\left(\boldsymbol{h}\middle|\boldsymbol{r},\boldsymbol{\Sigma}_v^{(\ell)}\right)\log p(\boldsymbol{h}\mid\boldsymbol{r},\boldsymbol{\Sigma}_v)d\boldsymbol{h}
\end{align}
is the cross-entropy between the last-iteration estimate and the current one. By Gibbs' inequality, $\log p(\boldsymbol{r}\mid\boldsymbol{\Sigma}_v)$ is lower bounded by $Q\left(\boldsymbol{\Sigma}_v\middle|\boldsymbol{\Sigma}_v^{(\ell)}\right)$ up to an additive constant independent of $\boldsymbol\Sigma_v$ \cite{Moon1996EM}. Hence, we can indirectly maximize $\log p(\boldsymbol{r}\mid\boldsymbol{\Sigma}_v)$ by maximizing $Q\left(\boldsymbol{\Sigma}_v\middle|\boldsymbol{\Sigma}_v^{(\ell)}\right)$. As $p(\boldsymbol{r},\boldsymbol{h}\mid\boldsymbol{\Sigma}_v) = p(\boldsymbol{r}\mid\boldsymbol{h},\boldsymbol{\Sigma}_v)p(\boldsymbol{h})$, the problem further simplifies to
\begin{align}
    \boldsymbol{\Sigma}_v = \argmax_{\boldsymbol{\Sigma}_v}\tilde{Q}\left(\boldsymbol{\Sigma}_v\middle|\boldsymbol{\Sigma}_v^{(\ell)}\right),
    \label{eq:Q2}
\end{align}
where 
\begin{align}
    \tilde{Q}\left(\boldsymbol{\Sigma}_v\middle|\boldsymbol{\Sigma}_v^{(\ell)}\right) = \E_{p\left(\boldsymbol{h}\middle|\boldsymbol{r},\boldsymbol{\Sigma}_v^{(\ell)}\right)}[\log p(\boldsymbol{r}\mid\boldsymbol{h},\boldsymbol{\Sigma}_v)].
\end{align}
Given that $\boldsymbol{v}\sim\mathcal{CN}(\boldsymbol{0},\boldsymbol{\Sigma}_v)$, we have
\begin{align}
    \tilde{Q}\left(\!\boldsymbol{\Sigma}_v\middle|\boldsymbol{\Sigma}_v^{(\ell)}\!\right) {=} {-}\tr(\boldsymbol{\Sigma}_v^{-1}\boldsymbol{C}_h){-}\log\det(\boldsymbol{\Sigma}_v){-}N_r\log\pi,
\end{align}
where
\begin{align}\label{eq:Ch}
    \boldsymbol{C}_h=\E_{p\left(\boldsymbol{h}\middle|\boldsymbol{r},\boldsymbol{\Sigma}_v^{(\ell)}\right)}\left[(\boldsymbol{r}-\boldsymbol{B}\boldsymbol{h})(\boldsymbol{r}-\boldsymbol{B}\boldsymbol{h})^\herm\right].
\end{align}

To comply with the OAMP principle, we project $\boldsymbol{\Sigma}_v$ onto a white noise model to pass a scalar noise variance. This projection is supported by the error orthogonality provided by OAMP, which renders $\boldsymbol{v}$ asymptotically i.i.d. With such an isotropic noise constraint $\boldsymbol{\Sigma}_v=\sigma_v^2\boldsymbol{I}$, \eqref{eq:Q2} becomes
\begin{align}
    \sigma_v^2 &= \argmax_{\sigma_v^2}\left({-}\frac{1}{\sigma_v^2}\tr(\boldsymbol{C}_h){-}N_r\log\sigma_v^2{-}N_r\log\pi\right),
    \nonumber\\
    &=\frac{1}{N_r}\tr(\boldsymbol{C}_h).
    \label{eq:sigma2v}
\end{align}

Now we can derive the closed-form expression for $\sigma_v^2$ by incorporating the variational distribution of $\boldsymbol{s}$ and $\boldsymbol{h}$ from OAMP, given in the following proposition.
\begin{proposition}\label{prop:sigma2v}
    When the signal matrix $\boldsymbol{B}(\tilde{\boldsymbol{s}}_\alpha)$ is approximated by the NLE message $\tilde{\boldsymbol{s}}_\alpha$ of the OAMP data detector, the variance of the effective white noise $\boldsymbol{v}$ in \eqref{eq:rBhv} is given by
    \begin{align}\label{eq:sigma2v_final}
        \sigma_v^2=\sigma_w^2+\frac{\chi_s^2}{N_r}\left(\chi_h^2N_hN_s + \tilde{\boldsymbol{h}}_\alpha^\herm\left(\sum_q\boldsymbol{B}_q^\herm\boldsymbol{B}_q\right)\tilde{\boldsymbol{h}}_\alpha\right),
    \end{align}
    where $\boldsymbol{B}_q=\boldsymbol{B}\left(\boldsymbol{e}_q^s\right)$ and $\boldsymbol{e}_q^s\in\C^{N_s}$ is the one-hot vector with its $q$-th element being $1$ and $0$ otherwise.
\end{proposition}
\begin{IEEEproof}
    The proof is given in Appendix \ref{app:sigma2v}.
\end{IEEEproof}

\begin{algorithm}[t]
    \caption{OAMP-JED}
    \label{alg:oamp_jed}
    \begin{algorithmic}[1]
        \Input $\boldsymbol{r},\boldsymbol{x}_\mathrm{p},P,\sigma_w^2,N_\text{ite}$
        \Output $\hat{\eta}_s\bigl(\boldsymbol{F}_{\!\mathrm{mod}}^\herm\tilde{\boldsymbol{s}}_{\beta,\mathrm{d}}\bigr)$
        \Initialize $\tilde{\boldsymbol{h}}_\alpha{\gets}\boldsymbol{0},\ \tilde{\boldsymbol{x}}_\alpha{\gets}\boldsymbol{x}_\mathrm{p},\ \chi_s^{2}{\gets}E_d,\ \chi_h^{2}{\gets}\frac{1}{N_{\!h}}$
        \For{$\ell=0,\ldots,N_\text{ite}-1$}
            \CmtLine{Module A: channel estimation}
            \State Construct $\boldsymbol{B}$ from $\tilde{\boldsymbol{s}}_\alpha=\boldsymbol{F}_{\!\mathrm{mod}}\tilde{\boldsymbol{x}}_\alpha$ using \eqref{eq:B}.
            \State Update $\sigma_v^2$ by \eqref{eq:sigma2v_final}.
            \State LE for channel: Update $\tilde{\boldsymbol{h}}_\beta$ by \eqref{eq:oamp_ce_le} and $\phi_h^2$ by \eqref{eq:phi2_h}.
            \State NLE for channel: Update $\tilde{\boldsymbol{h}}_\alpha$ by \eqref{eq:oamp_ce_nle} and $\chi_h^2$ by \eqref{eq:chi2_h}.
            \CmtLine{Module B: data detection}
            \State Construct $\boldsymbol{G}$ from $\tilde{\boldsymbol{h}}_\alpha$ using $\boldsymbol{G}=\tilde{\boldsymbol{G}}\boldsymbol{F}_{\mathrm{zp}}$ and \eqref{eq:G}.
            \State Update $\sigma_u^2$ by \eqref{eq:sigma2u_final}.
            \State LE for data: Update $\tilde{\boldsymbol{s}}_\beta$ by \eqref{eq:oamp_dd_le} and $\phi_s^2$ by \eqref{eq:phi2_s}.
            \State Pilot removal: $\tilde{\boldsymbol{s}}_{\beta,\mathrm{d}}\gets\tilde{\boldsymbol{s}}_{\beta}-\boldsymbol{F}_{\!\mathrm{mod}}\boldsymbol{x}_\mathrm{p}$.
            \State NLE for data: Update $\tilde{\boldsymbol{x}}_{\alpha,\mathrm{d}}$ by \eqref{eq:oamp_dd_nle} and $\chi_s^2$ by \eqref{eq:chi2_s}.
            \State Pilot addition: $\tilde{\boldsymbol{x}}_\alpha\gets\tilde{\boldsymbol{x}}_{\alpha,\mathrm{d}}+\boldsymbol{x}_\mathrm{p}$.
        \EndFor
    \end{algorithmic}
\end{algorithm}

Next, we derive the variational noise term $\boldsymbol{u}$ in \eqref{eq:rGsu} for Module~B. Following the same derivation, we perform the white noise projection $\boldsymbol{\Sigma}_u=\sigma_u^2\boldsymbol{I}$ and establish the following proposition.
\begin{proposition}\label{prop:sigma2u}
    When the channel matrix $\boldsymbol{G}(\tilde{\boldsymbol{h}}_\alpha)$ is approximated by the NLE message $\tilde{\boldsymbol{h}}_\alpha$ of the OAMP channel estimator, the variance of the effective white noise $\boldsymbol{u}$ in \eqref{eq:rGsu} is given by
    \begin{align}\label{eq:sigma2u_final}
        \sigma_{u}^2=\sigma_w^2+\frac{\chi_h^2N_h}{N_r}\left(\chi_s^2N_s + \tilde{\boldsymbol{s}}_\alpha^\herm\tilde{\boldsymbol{s}}_\alpha\right).
    \end{align}
\end{proposition}
\begin{IEEEproof}
    The proof is given in Appendix \ref{app:sigma2u}.
\end{IEEEproof}

Intuitively, $\sigma_v^2$ and $\sigma_u^2$ can be viewed as the variances of the effective white noise observed by the OAMP estimators, which is Bayesian-optimal for the variational Gaussian error implied by the current channel and signal estimates \cite{Wang2022OTFS_JED_VBI}. This scalarization also matches OAMP's asymptotic isotropic error \cite{Ma2017OAMP}. Based on the above derivations, we summarize the proposed OAMP-JED algorithm in Algorithm \ref{alg:oamp_jed}.

\subsection{Computational Complexity}\label{sec:complexity}
We summarize the computational complexity of the proposed OAMP-JED receiver for the uncoded case. For coded ODDM, additional decoder-specific complexity is incurred.

The complexities of Module~A and Module~B are dominated by the matrix inversions in \eqref{eq:WHat_B} and \eqref{eq:WHat_G}, respectively. Using the Woodbury identity, \eqref{eq:WHat_B} has complexity $\mathcal{O}(N_rN_h^2)$, which is affordable due to the small $N_h$ of sparse doubly selective channels. A direct computation of \eqref{eq:WHat_G} costs a prohibitive $\mathcal{O}(N_rN_s^2)$ for typical ODDM frame sizes. However, with the ZP structure, $\boldsymbol{G}$ is block diagonal with block size $\C^{M\times M_d}$, so \eqref{eq:WHat_G} can be evaluated blockwise, reducing the complexity to $\mathcal{O}(NM_d^2(M+M_d))$. For the noise inflation, both \eqref{eq:sigma2v_final} and \eqref{eq:sigma2u_final} incur negligible complexity. Specifically, $\sum_q\boldsymbol{B}_q^\herm\boldsymbol{B}_q$ in \eqref{eq:sigma2v_final} is independent of signal and channel realizations, which can be precomputed once per transmission setting and reused. \eqref{eq:sigma2u_final} is a simple scalar computation. Therefore, the overall complexity of OAMP-JED is $\mathcal{O}\bigl(N_{\text{ite}}\bigl(N_rN_h^2+NM_d^2(M+M_d)\bigr)\bigr)$, where $N_{\text{ite}}$ is the number of iterations.

Table~\ref{tab:complexity_compare} compares the uncoded complexity orders of OAMP-JED and two benchmark JED receivers. For MF-OAMP \cite{Wen2024MFOAMP}, $I_{\mathrm{est}}$ and $I_{\mathrm{det}}$ denote the numbers of nested OAMP iterations for channel estimation and data detection, respectively. For SP-I \cite{Mishra2022OTFS_SuperimposedPilots}, $I_{\mathrm{MP}}$ is the number of nested detection iterations. Without nested iterations, OAMP-JED has a complexity order comparable to that of SP-I and lower than that of MF-OAMP.

\begin{table}[t]
    \centering
    \caption{Computational complexity of uncoded JED receivers.}
    \renewcommand{\arraystretch}{1.2}
    \begin{tabular}{|c|c|}
        \hline
        Receiver & Computational complexity \\
        \hline
        Proposed OAMP-JED &
        $\mathcal{O}\bigl(N_{\text{ite}}\bigl(N_rN_h^2+NM_d^2(M+M_d)\bigr)\bigr)$\\
        \hline
        SP-I \cite{Mishra2022OTFS_SuperimposedPilots} &
        $\mathcal{O}\bigl(N_{\text{ite}}\bigl(P^2N_r{+}P^3{+}I_{\mathrm{MP}}N_r P\abs{\mathcal{A}}\bigr)\bigr)$\\
        \hline
        MF-OAMP \cite{Wen2024MFOAMP} &
        $\mathcal{O}\bigl(N_{\text{ite}}\bigl(I_{\mathrm{est}}N_rN_h^2+I_{\mathrm{det}}N_rM_d^2(M{+}M_d)\bigr)\bigr)$\\
        \hline
    \end{tabular}
    \label{tab:complexity_compare}
\end{table}

\section{Numerical Results}\label{sec:numerical}

In this section, we present numerical results for the proposed OAMP-JED receiver in ODDM systems. An ODDM signal is transmitted with $M{=}256$, $N{=}32$, and $T{=}\SI{66.67}{\micro\second}$, using QPSK modulation. The channel follows the TDL-A power delay profile with $P{=}23$ paths \cite{3GPPTR38901R19}, while Doppler shifts are generated according to the Jakes' model. We set the carrier frequency to \SI{5}{\giga\hertz} and the maximum velocity to \SI{360}{\kilo\metre/\hour}. Off-grid delay and Doppler shifts are rounded to the nearest integer. Based on the parameters, we have $l_{\max}{=}10,k_{\max}{=}4$. Therefore, the adopted ZP structure incurs an overhead of $10/256\approx 3.9\%$ in the considered scenario, which is comparable to the CP overhead in 5G NR \cite{3GPPTS38211R19}. For coded scenarios, we employ a $(15744,8448)$ 5G NR LDPC code (base graph 1), decoded by a belief propagation decoder with at most $20$ decoding iterations.

For comparison, we consider two benchmark JED receivers, SP-I \cite{Mishra2022OTFS_SuperimposedPilots} and MF-OAMP \cite{Wen2024MFOAMP}. To quantify the loss due to imperfect CSI, we further include two genie-aided perfect-CSI references: OAMP \cite{Ma2017OAMP} and the parallel interference cancellation with MMSE (PIC-MMSE) detector based on \cite{Li2025OTFS_ParallelCoding}. The iteration limit is set to $10$ for SP-I, OAMP-JED, and OAMP. For MF-OAMP, we use $10$ outer iterations and $5$ inner iterations for the nested OAMP loops. These iteration settings have been numerically verified to be sufficient for all algorithms to converge. The superimposed pilot sequence $\boldsymbol{x}_\mathrm{p}$ for JED algorithms is generated as an i.i.d. Gaussian sequence.\footnote{
    In practice, low-PAPR sequences, such as Zadoff-Chu sequences, can also be adopted as pilots without affecting the proposed OAMP-JED derivation.
}
We define $E_b/N_0$ based on the total bit energy including the power allocated to $\boldsymbol{x}_\mathrm{p}$. 

We first examine the impact of the pilot-to-data energy ratio $\gamma$ on the BER performance of JED algorithms, as shown in Fig. \ref{fig:ber_gamma}. Each JED algorithm exhibits an optimal $\gamma$ that minimizes the BER by balancing the pilot SNR for channel estimation and the data SNR for detection. The proposed OAMP-JED receiver consistently outperforms SP-I and MF-OAMP, with the optimal $\gamma$ around \SI{-12}{\decibel} for both uncoded and coded systems. This optimal $\gamma$ is higher than those of SP-I and MF-OAMP, suggesting that OAMP-JED better exploits the data detection gain. Moreover, OAMP-JED closely approaches the perfect-CSI OAMP benchmark at the optimal $\gamma$, confirming the effectiveness of OAMP-JED. In the following, we adopt the respective optimal $\gamma$ for each JED algorithm. 

\begin{figure}[t]
    \centering
    \subfloat[Uncoded ODDM with $E_b/N_0=\SI{12}{\decibel}$.]{\label{fig:berGAMMA_qpsk_uncoded}\includegraphics[width=0.99\columnwidth,trim={0 0 0 0},clip]{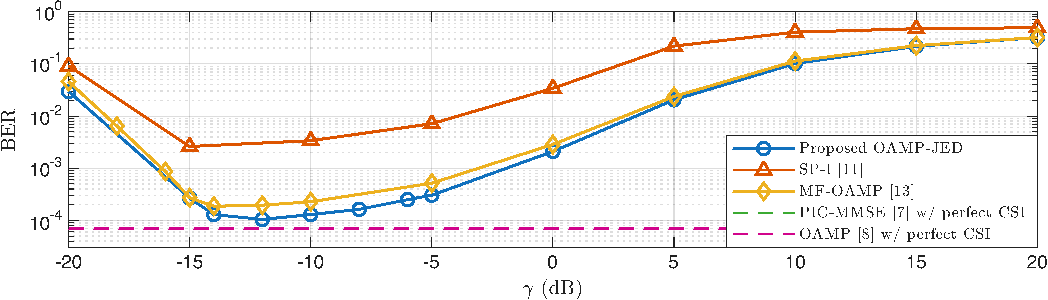}}
    \\
    \centering
    \subfloat[Coded ODDM with $E_b/N_0=\SI{8}{\decibel}$.]{\label{fig:berGAMMA_qpsk_coded}\includegraphics[width=0.99\columnwidth,trim={0 0 0 0},clip]{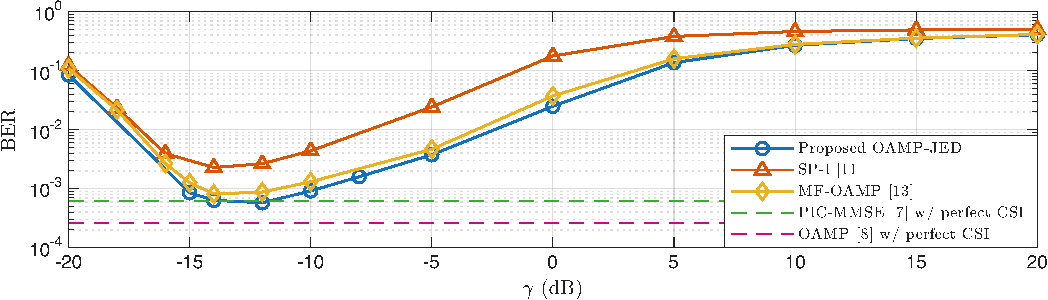}}
    \caption{BER versus $\gamma$ for ODDM.}
    \label{fig:ber_gamma}
\end{figure}

We next evaluate the proposed OAMP-JED receiver for uncoded ODDM, as shown in Fig. \ref{fig:ber_snr_uncoded}. Thanks to OAMP's error orthogonality, OAMP-JED better controls error propagation and achieves a significantly lower BER than SP-I. Although the recently proposed MF-OAMP tracks OAMP-JED at low $E_b/N_0$, it exhibits an error floor at high $E_b/N_0$ due to the suboptimal MF approximation. Moreover, OAMP-JED closely approaches the performance of OAMP and PIC-MMSE with perfect CSI, maintaining a small gap of less than $\SI{0.2}{\decibel}$. This indicates that the gap is mainly due to the reduced data SNR caused by the superimposed pilot.

\begin{figure}[ht]
    \centering
    \subfloat[Uncoded ODDM.]{\label{fig:ber_snr_uncoded}\includegraphics[width=0.48\columnwidth,trim={0 0 0 0},clip]{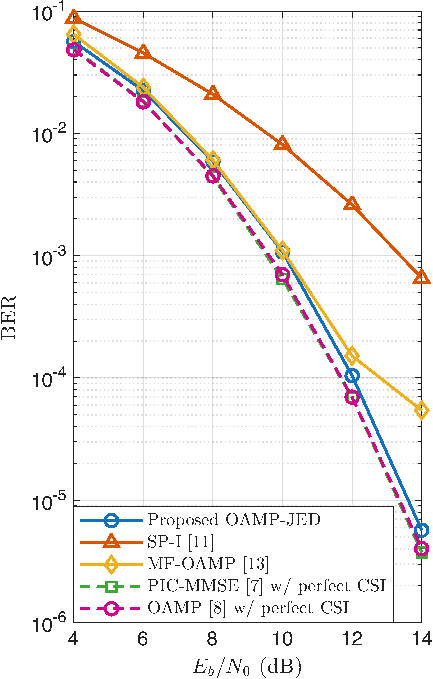}}
    \hfill
    \subfloat[Coded ODDM.]{\label{fig:ber_snr_coded}\includegraphics[width=0.48\columnwidth,trim={0 0 0 0},clip]{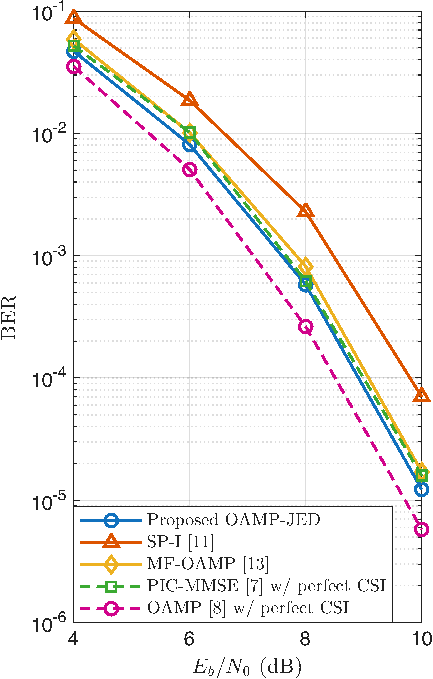}}
    \caption{BER versus $E_b/N_0$ for ODDM. The selected $\gamma$ for (OAMP-JED, SP-I, MF-OAMP) are $(-12,-16,-14)\,\si{\decibel}$ in the uncoded case and $(-12,-14,-14)\,\si{\decibel}$ in the coded case.}
    \label{fig:ber_snr}
\end{figure}

Finally, we evaluate the proposed OAMP-JED receiver for coded ODDM, illustrated in Fig. \ref{fig:ber_snr_coded}. Among the benchmark JED algorithms, OAMP-JED still achieves the lowest BER. Although MF-OAMP now approaches OAMP-JED across the simulated $E_b/N_0$ range, it incurs a significantly higher computational complexity due to its nested OAMP loops. OAMP-JED also closely approaches OAMP under perfect CSI, with the gap slightly enlarged to about $\SI{0.4}{\decibel}$, since the reduced data SNR translates to a larger BER degradation after decoding. Interestingly, with channel coding, OAMP-JED outperforms PIC-MMSE with perfect CSI, which we attribute to the enhanced error orthogonality provided by the OAMP principle \cite{Ma2017OAMP}, a property favored by the decoder.

\section{Conclusion}

In this paper, we proposed an OAMP-JED receiver for ODDM over doubly selective channels. By leveraging the OAMP framework, the proposed receiver efficiently performs joint channel estimation and data detection with a scalar noise variance update to account for model uncertainty. OAMP's error orthogonality enables principled soft information exchange across iterations with reduced error propagation. Numerical results demonstrate that the proposed OAMP-JED receiver outperforms benchmark JED algorithms and closely approaches the performance of OAMP with perfect CSI.

\appendices

\section{Proof of Proposition \ref{prop:sigma2v}}\label{app:sigma2v}

We construct the approximate signal matrix $\boldsymbol{B}(\tilde{\boldsymbol{s}}_\alpha)$ using the NLE message from the OAMP data detector. Substituting it into \eqref{eq:rBh} leads to
\begin{align}\label{eq:r-Bh}
    \boldsymbol{r}-\boldsymbol{B}(\tilde{\boldsymbol{s}}_\alpha)\boldsymbol{h}=\boldsymbol{w}-\boldsymbol{B}(\Delta\boldsymbol{s})\boldsymbol{h},
\end{align}
with $\Delta\boldsymbol{s}=\tilde{\boldsymbol{s}}_\alpha-\boldsymbol{s}$. Combining with \eqref{eq:Ch} yields
\begin{align}
    \boldsymbol{C}_h=\cov(\boldsymbol{B}(\Delta\boldsymbol{s})\boldsymbol{h})+\sigma_w^2\boldsymbol{I},
    \label{eq:Ch2}
\end{align}
where $\cov(\boldsymbol{B}(\Delta\boldsymbol{s})\boldsymbol{h})=\E\left[\boldsymbol{B}(\Delta\boldsymbol{s})\E\left[\boldsymbol{h}\boldsymbol{h}^\herm\right]\boldsymbol{B}^\herm(\Delta\boldsymbol{s})\right]$. Linearity implies $\boldsymbol{B}(\Delta\boldsymbol{s})=\sum_{q}\Delta s[q]\boldsymbol{B}_q$, yielding
\begin{align}\nonumber
    \cov(\boldsymbol{B}(\Delta\boldsymbol{s})\boldsymbol{h})=\sum_{q_1,q_2}\E[\Delta s[q_1]\Delta s^*[q_2]]\boldsymbol{B}_{q_1}\E\left[\boldsymbol{h}\boldsymbol{h}^\herm\right]\boldsymbol{B}_{q_2}^\herm.
\end{align}
Given OAMP's divergence-free NLE, $\Delta\boldsymbol{s}$ behaves like white noise with scalar variance $\chi_s^2$. Therefore, we have
\begin{align}
    \cov(\boldsymbol{B}(\Delta\boldsymbol{s})\boldsymbol{h})=\chi_s^2\sum_{q}\boldsymbol{B}_q\E\left[\boldsymbol{h}\boldsymbol{h}^\herm\right]\boldsymbol{B}_q^\herm.
\end{align}
Similarly, we use the NLE message from the OAMP channel estimator to construct the surrogate distribution, i.e., $p\left(\boldsymbol{h}\middle|\boldsymbol{r},\boldsymbol{\Sigma}_v^{(\ell)}\right)=\mathcal{CN}\left(\tilde{\boldsymbol{h}}_\alpha,\chi_h^2\boldsymbol{I}\right)$. This yields
\begin{align}\label{eq:CBh}
    \cov(\boldsymbol{B}(\Delta\boldsymbol{s})\boldsymbol{h}){=}\chi_s^2\sum_{q}\left(\chi_h^2\boldsymbol{B}_q\boldsymbol{B}_q^\herm{+}\boldsymbol{B}_q\tilde{\boldsymbol{h}}_\alpha\tilde{\boldsymbol{h}}_\alpha^\herm\boldsymbol{B}_q^\herm\right).
\end{align}
Substituting \eqref{eq:Ch2} and \eqref{eq:CBh} back to \eqref{eq:sigma2v} gives
\begin{align}
    \sigma_v^2=\sigma_w^2{+}\frac{\chi_s^2}{N_r}\left(\!\chi_h^2\sum_q\tr\!\left(\!\boldsymbol{B}_q\boldsymbol{B}_q^\herm\!\right){+}\tilde{\boldsymbol{h}}_\alpha^\herm\left(\!\sum_q\boldsymbol{B}_q^\herm\boldsymbol{B}_q\!\right)\tilde{\boldsymbol{h}}_\alpha\!\right).
    \nonumber
\end{align}
According to the construction of $\boldsymbol{B}(\boldsymbol{s})$ in \eqref{eq:B}, we have $\tr\left(\boldsymbol{B}_q\boldsymbol{B}_q^\herm\right)=N_h$ for $q=0,\dots,N_s-1$. Thus, we arrive at \eqref{eq:sigma2v_final}, which completes the proof.

\section{Proof of Proposition \ref{prop:sigma2u}}\label{app:sigma2u}
Following the same derivation as in Appendix \ref{app:sigma2v}, we reach
\begin{align}\label{eq:sigma2u}
    \sigma_{u}^2=\sigma_w^2+\frac{\chi_h^2}{N_r}\left(\chi_s^2N_hN_s + \tilde{\boldsymbol{s}}_\alpha^\herm\left(\sum_i\boldsymbol{G}_i^\herm\boldsymbol{G}_i\right)\tilde{\boldsymbol{s}}_\alpha\right),
\end{align}
where $\boldsymbol{G}_i=\boldsymbol{G}\left(\boldsymbol{e}_i^h\right)$ and $\boldsymbol{e}_i^h\in\C^{N_h}$ is the one-hot vector with its $i$-th element being $1$ and $0$ otherwise. Moreover, $\boldsymbol{G}_i$ has orthonormal columns, so $\boldsymbol{G}_i^\herm\boldsymbol{G}_i=\boldsymbol{I}$. Therefore, \eqref{eq:sigma2u} can be simplified to \eqref{eq:sigma2u_final}, which completes the proof.


\ifCLASSOPTIONcaptionsoff
  \newpage
\fi

\bibliographystyle{IEEEtran}
\bibliography{references}

\end{document}